\documentclass[journal]{IEEEtran}
\usepackage{graphicx}
\usepackage[normalem]{ulem}
\usepackage{amssymb}
\usepackage{amsmath}
\usepackage{amstext}
\usepackage{amsfonts}
\usepackage[latin1]{inputenc}
\usepackage[normalem]{ulem}
\usepackage[brazil]{babel}

\newtheorem{definition}{Definition}
\newtheorem{proposition}{Proposition}
\newtheorem{corollary}{Corollary}
\newtheorem{proof}{Proof}

\begin{document}
\renewcommand*\abstractname{Abstract}

\title{Shannon and Renyi Entropy of Wavelets}

\author{H. M. de Oliveira, \\
Statistics Department\\
Federal University of Pernambuco\\
CCEN-UFPE, Recife, Brazil}

\maketitle

\begin{abstract}

\textbf{This paper reports a new reading for wavelets, 
which is based on the classical 'De Broglie' principle. 
The wave-particle duality principle is adapted to wavelets. 
Every continuous basic wavelet is associated with a proper 
probability density, allowing defining the Shannon entropy 
of a wavelet. Further entropy definitions are considered, 
such as Jumarie or Renyi entropy of wavelets. 
We proved that any wavelet of the same family 
$\{ {\psi_{a,b} (t)} \}_{a \neq 0}$  has the same Shannon entropy 
of its mother wavelet. Finally, the Shannon entropy for a few 
standard wavelet families is determined.} \footnote{Pub: Int. J. of Mathematics and Computer Science, Vol.10 (2015), no. 1, 13-26, e-mail:hmo@de.ufpe.br. \\
This work was partially supported by the Brazilian National Council for Scientific and Technological Development (CNPq) under research grant 306180.}
\end{abstract}

\providecommand{\keywords}[1]{\textbf{\textit{Index terms---}} #1}

\section{PRELIMINARIES AND BACKGROUND}
Little has been made for analogue signals in the information theory scope as compared to the amazing coverage nowadays available for digital signals  \cite{CoverThomas}. 
This paper is precisely focused on this rather unexplored field, taking advantage of a fresh and powerful tool: the wavelet analysis 
\cite{Vetterli},  \cite{de Oliveira} - that evolved into a specialised branch of the modern-day signal processing. 90's witnessed the emergence of wavelets, which rapidly reached in practice, 
thanks to their natural feature of concentrating energy in a few transform coefficients. This paper intends to introduce a new insight into wavelets, 
which is based on the conventional De Broglie duality principle and the statistical interpretation of the wave-function formulated by Max Born  \cite{Beiser}. 
Two concepts of information (logon and Shannon information) are considered  \cite{Gabor},  \cite{Shannon}, as well as their relation with entropy. 
Entropy (Greek: $en+trope$=in+turning) is one of the most fundamental concepts of Science. Since the notion of entropy appeared, 
it has always been surrounded by a halo of inscrutability. The well-known German chemist W. Ostwald put it in this way: 
"Energy is the queen of the world, and entropy is her shadow!" \  It was also told that when Von Neumman suggested that Shannon use the word entropy, 
he added, "it will give you a great edge in debates because nobody really knows what entropy is anyway"  \ \cite{Bricmont}. 
Quite often people hear about entropy the first time when the most tantalising problems such as the origin of the life or the future of the universe is discussed. 
Schr\"odinger has mentioned \cite{Schrodinger} that living organisms feed on negative entropy, i.e., they drive in the direction of increasing organisation. 
The proposal of this paper is mainly to build a bridge between two areas of applied mathematics: the analysis-decomposition (with wavelets) and the Information Theory, from a perspective not much explored so as to awaken interest in new advances. It should be clarified that techniques and results derived in this article are naive. \\
\\
To begin with, a first question is raised: "Is it possible to associate an entropy measure with a wavelet"?
 Jumarie  \cite{Jumarie} introduced entropy associated with a given continuous differentiable function $f:\Omega \subseteq \mathbb{R}\rightarrow \mathbb{R}$ as ($\mathbb{R}$ denotes the real set)

\begin{equation}
H(f(.);\Omega )):=\frac{\int_{\Omega }^{} |f'(x)|log |f'(x))dx}{\int_{\Omega }^{} |f'(x)|dx}
\end{equation}

As a first attempt, the entropy of a continuous differentiable wavelet $\psi: \mathbb{R} \rightarrow  \mathbb{R}$  can be defined by
\begin{equation}
H(\psi ,\mathbb{R}):=\frac{\int_{-\infty }^{+\infty} |\psi '(t)|.log|\psi '(t)|dt}{\int_{-\infty }^{+\infty} |\psi '(t)|dt}
\end{equation}

The seminal Max Born footnote for interpreting the solution of Schr\"odinger equation in quantum mechanics is applied here in the wavelet framework. 
As a consequence, it is suggested that wavelets can behave as some kind of little particle (corpuscle of a true random nature). 
Bearing in mind that the square of the wave-function is a probability density, we propose in a parallel way to associate a probability 
density function (PDF) $p_t (\psi):= \psi^2 (t)$ with every basic continuous wavelet $\psi (t)$. 
Given that Fourier transform is an isometric transform, we can go further anchored in the Parseval identity and propose to associate 
$\psi$ with an additional density function expressed by $p_f (\psi) := \frac{1}{2 \pi} | \Psi (w)|^2$, 
where $\Psi (w)$ is the Fourier transform of $\psi (t)$. The fundamental challenge is to determine the behaviour of a corpuscle 
when its freedom of motion is limited by the action of external forces: each wavelet describes a specific situation. 
It can be therefore stated that: "wavelets are to corpuscle as wave-functions are to particles."  \
After unveiling such probabilistic properties associated with wavelets, it is intuitive to set up another concept: 
the Shannon entropy associated with a wavelet as a measure of the disorder of a signal. 
The entropy of a random variable can be defined in the discrete case as well as in the continuous case \cite{Kolmogorov}, \cite{Battail}. 
In the later case, the so-called Shannon differential entropy of a random variable $X$ with probability density $p(x)$ is defined by
\begin{equation}
H(X):=-\int_{-\infty }^{+\infty } p(x).log p(x)dx
\end{equation}

The information unit depends on the base of the logarithm. For the sake of convenience, shannon (binary unit) is adopted 
through this paper for the information unity.
\footnote{Although the attempt to replace the term bit by the term shannon unit, used by the International Standard Organization (ISO) in 1975, retrospectively had not been very successful.}
 
According to the above vindication, the entropy of a wavelet can be measured by:

\begin{definition}
(Shannon entropy of a wavelet). The time entropy, $H_t (\psi)$, of a continuous wavelet $\psi(.)$ is defined by
\begin{equation}
H_{t}(\psi ):=-\int_{-\infty }^{+\infty }\psi ^{2}(t).log_{2}\psi ^{2}(t)dt
\end{equation}
\end{definition}

In an parallel way, the frequency entropy, $H_f (\psi)$, of a continuous wavelet $\psi(.)$ is defined by

\begin{equation}
H_{f}(\Psi  ):=-\int_{-\infty }^{+\infty }\frac{1}{2\pi }|\Psi (w)|^2.log_{2}\frac{1}{2\pi }|\Psi (w)|^{2}dw
\end{equation}

The entropy gives information on the spreading of the wavelet, i.e., it furnishes a "localising measure" 
of the corpuscle in a particular domain (time or frequency). 
The probability distribution function associated with the density $p_t (\psi)$ is given by $P(t)=\int_{-\infty }^{t}\psi^{2}(t')dt')$. 
The time Shannon entropy of a wavelet is, therefore, exactly the Jumarie entropy of the PDF related to the density 
$\psi^2 (t)$, i.e., $H_t (\psi)=H(P(.),\mathbb{R})$. 
For completeness, other "non-shannonian" \ measures such as the Renyi entropy of order $s > 0$ could promptly be defined for wavelets \cite{Battail}:

\begin{equation}
H_t (\psi | s):=\frac{1}{s-1}ln \left ( \int_{-\infty }^{+\infty }\psi ^{2s}(t')dt' \right ) 
\end{equation}
which hold 
\begin{equation}
\lim_{s\to 1} H_{t}(\psi|s)=H_{t}(\psi)
\end{equation}

Indeed, while Gabor functions can be used to derive similar results, it is worth noting that standard are also wavelets. 
The harmonic oscillator equation solutions \cite{Beiser} have these properties.  In fact, Hermite polynomials of degree $n, H_n(t)$, 
modulating a Gaussian pulse, yield either a wavelet or a scale function (respectively, $n$ odd, $n$ even). 
Now, even further class of functions such as 'wave functions' of the Schr\"{o}dinger equation could also be considered, 
it must be remembered that they can also be interpreted as wavelets. \\

Conducting a careful literature search, attempts to connect wavelets and entropy were found. 
In 1999, Quian Quiroga, Rosso and Basar have successfully applied a disorder measure in neuroscience \cite{Quiroga1}. 
Here, the total energy $E_{tot}$ of the signal in each time window is calculated as the sum of energies of all resolution levels. 
The relative wavelet energy $P_j$ is computed as the ratio between the energy of each level, 
$E_j$, and total energy of the signal, $E_{tot},$ in the respective time window. 
The wavelet entropy, $S_{WT}$, was then defined by

\begin{equation}
S_{WT}:=-\sum_{j}^{}P_{j}lnP_{j}
\end{equation}

The main focus of this tool had specifically been on the electroencephalogram analysis \cite{Quiroga2}, \cite{Rosso}, \cite{Yordanova}, \cite{Hornero}. 
Thereafter, this concept was applied to astronomy with the aim of investigating the solar activity \cite{Sello}. 
Despite the fact that $S_{WT}$ had been referred to as the wavelet entropy up to now, it should not be named so. 
Actually, $S_{WT}$ depends on the analysed signal and it is not a solely feature of the $\psi$ $wavelet$ $itself$. 
The term "wavelet entropy" is thereby somewhat inappropriate, 
and it should be better called as "the entropy of a wavelet decomposition of a signal."

\section {On The Shannon Entropy of Continuous Wavelets}

The effect of scaling or shifting a mother wavelet on the time and frequency entropy is initially examined.

\begin{proposition}:  
Given a continuous mother wavelet $\psi(.)$ with time entropy $H_t (\psi)$ and frequency entropy $H_f (\psi)$, 
the entropy of a daughter wavelet ${\psi}_{a,b} (t):= \frac{1} {\sqrt{|a|}} {\psi} ( \frac{t-b}{a})$, $a\neq 0$ 
can be computed by $H_t (\psi_{a,b} )= H_t (\psi) + log_2 |a|$, 
and $H_f (\psi_{a,b})= H_f (\psi) - log_2 |a|$. 
\end{proposition}

\begin{proof}. The first part follows from substituting $\psi^2_{a,b} = \frac{1}{|a|} \psi^2 ( \frac{t-b} {a})$ in Definition 1, 
and evoking that both $\psi(.)$ and $\psi_{a,b}(.)$ have normalised energy. 
The second part is derived using $| \Psi_{a,b} (w) |^{2} = |a|.| \Psi (aw)|^2$.
\end{proof}

The highest time entropy among all supportly compacted wavelets is achieved by the Haar wavelet. 
This is in agreement with the fact that maximum entropy of a discrete random variable 
is achieved by a uniform distribution as it can be seen by the following proposition:

\begin{proposition}:  
The time entropy of any wavelet of compact support is bounded by 
$H_{t}\psi \leq log_2 length(Supp(\psi))$, and the bound is only met by the Haar wavelet.
\end{proposition}

\begin{proof}. 
Let $L:=length(Supp(\psi))$ denote the length of the support of the wavelet and 
$\Delta:=H_t (\psi) -log_2 L$. Clearly
\begin{equation}
\Delta :=\int_{Supp\psi}{} \psi^{2}(t)log_{2} \frac{1}{L \psi^2(t)} dt
\end{equation}

Then
\begin{equation}
\Delta \leqslant \frac{1}{ln2}\int_{Supp (\psi) }{} \psi ^{2}(t)\left ( \frac{1}{L\psi ^{2}(t))}-1 \right )\leqslant 0
\end{equation}

The upper bound is only achieved by a Haar scaled version $\psi_{a,b}^{Haar} (t)$, where $a=length(Supp(\psi))/2$.  
\\
\end{proof}

\begin{corollary}: The time Shannon entropy of dBN wavelet is bounded by $H_t (dBN) \leqslant log_2 (2N-1)$. 
\end{corollary}

This result can equally be translated into the frequency domain, deriving, for instance,

\begin{corollary}: 
The frequency Shannon entropy of the deO wavelet \cite{de Oliveira1} is upper bounded by $H_f (deO) \leqslant  log_2 (\pi + 3 \pi \alpha)$. 
\end{corollary}

A pertinent comment should now be pointed out. According to the deterministic approach, 
in the cases where the wavelet waveform is fully known, so is its spectrum. 
One could argue that there is no information in the frequency domain. 
Now, this reasoning is fallacious. If the wavelet spectrum is perfectly determined, 
no information is provided in the time domain! 
Instead of this, it seems to be some amount of information in both domains matching with the Gabor uncertainty principle. 
In the probabilistic interpretation offered in this paper, even knowing $\psi(.)$ and $\Psi(.)$, 
it do exists some nonzero uncertainty in both time and frequency domain. 
The global entropy would be partially due to the (inherent) time-uncertainty and partially due to the (inherent) frequency-uncertainty. 
The entropy is only added when dealing with independent variables. 
We assume that the mutual information between time and frequency domain is zero, i.e., 
$I(\psi,\Psi)= H_t (\psi) + H_f (\psi) - H_{t,f} (\psi)=0$. 

An interesting argument consists of defining the global entropy of a wavelet as the sum of the entropy in both domains. 
\\
\\
Formally,

\begin{definition}
(Global entropy of a continuous wavelet). The global entropy of a wavelet $\psi(.)$ is defined by
\begin{equation}
H_\psi := H_t (\psi)+ H_f (\psi)
\end{equation}
\end{definition}

A direct property follows from such a definition: every daughter wavelet has the same global entropy of the mother wavelet, 
some sort of conservation principle. It follows then

\begin{corollary}: 
The global entropy is preserved within the same wavelet family $\left \{ \psi _{a,b}(t) \right \}_{a\neq 0,b\epsilon \mathbb{R}}$ 
so we are able to find a unique entropy value associate to a wavelet basis.
\end{corollary}

Thermodynamic concepts of entropy are always related to the temperature. Particularly, common units for entropy are Joule.K$^{-1}$ and cal.K$^{-1}$. 
Having defined the entropy of wavelets, another concept of interest could be the "temperature of a wavelet". 
This modus operandi suggests computing the temperature as the ratio between energy and entropy of the wavelet, i.e., 
$T_ {\psi} := E_ {\psi} / H_{\psi}$, where $E_{\psi }:=\int_{-\infty }^{+\infty }\psi ^{2}(t)dt$ is the energy of the wavelet. 
For instance, the temperature of the complex Morlet wavelet is $T_{CMor}\approx 0.3232^{o}$, that is hotter than the mexican hat wavelet ($T_{Mexh}\approx 0.2700^{o}$), 
and Haar wavelet is even colder, $T_{Haar} \approx0.2006^{o}$.  

\begin{table}[!h]
\caption{ANALYTICAL EXPRESSIONS OF SOME CONTINUOUS WAVELETS: MORLET, MEXICAN HAT, GAUSS1, SHANNON, HAAR, AND DE OLIVEIRA WAVELETS.} \label{tab1}
\begin{tabular}{c c}
\hline \hline
Complex Morlet &$\frac{e^{j5t}.e^{-t^2/2}}{\sqrt[4 ]{\pi}}$\\
(CMor)&\\
Sombrero &$\frac{2}{\sqrt{3}}(t^{2}-1)\frac{e^{-t^{2}/2}}{\sqrt[4]{\pi }}$\\
(mexh)&\\
C Shannon&$sinc(t)e^{-j2\pi t}$, $sinc(t):=\frac{sin\pi t}{\pi t}$\\ 
(Sinc)&\\
Haar &$\left\{\begin{matrix}
\frac{1}{\sqrt{2}} & if \  0< t< 1\\ 
-\frac{1}{\sqrt{2}} & if  \ -1< t< 0 \\ 
0 & otherwise. 
 \end{matrix}\right.$\\
(Haar)&\\
Morlet &$\sqrt{2}cos(5t).\frac{e^{-t^{2}/2}}{\sqrt[4]{\pi }}$\\
(Mor)&\\
Gaussian 1 &$\frac{\sqrt{2}te^{-t^{2}/2}}{\sqrt[4]{\pi }}$\\
(gauss1)&\\
Shannon&$sinc(\frac{t}{2})cos(\frac{3\pi t}{2})$\\
(Sha)&\\
de Oliveira &$\Re e \ \psi (t,\alpha )=s(t-0.5,\alpha )$\\
(CdeO)&$\Im m \  \psi (t,\alpha )=\widetilde{s}(t-0.5,\alpha )$\\
&$s$ and $\widetilde{s}$ defined according to \cite{de Oliveira1}.\\
\hline
\end{tabular}
\end{table}

Since the two densities $\psi^2 (t)$ and $\frac{1}{2 \pi} | \Psi(w)|^2$ are related via the transform pair $\psi (t) \Leftrightarrow \Psi(w)$, the time-frequency relationship is disclosed on the "partition"  of the total uncertainty. 
The concept of isoresolution (introduced \cite{Soares}) can now be bringing into play. 
Specifically, "Are there wavelets that achieve the same entropy in both domain"? 
A class of signals that hold special and nice-looking properties are the invariant signals under the Fourier transform. 
Wavelets that belong to this specific class achieve matching time and frequency entropy, 
upholding the earlier idea of isoresolution. Evoking the following proposition can easily prove this \cite{Soares}:

\begin{proposition}:  
Possible eigenvalues of the Fourier transform operator are one of the four roots of the unit $(\pm 1,\pm j)$ times $\sqrt{2 \pi}$.
\end{proposition}

\begin{proof}: see \cite{Soares}. 
\end{proof}\

\begin{proposition}:  
 Isoresolution wavelets hold 
\begin{equation}
H_t (\psi) = H_f (\psi).
\end{equation}
\end{proposition}

\begin{proof}: Fourier eingenfunction wavelets possess $|\Psi (w)|= \sqrt {2 \pi} |\psi(w)|$ and the proof follows. 
\end{proof}\

In order to gain insight into this new thought on wavelets, a number of standard wavelets were selected (Table I) 
and we plotted both time and frequency density functions associated with such wavelets (Figure 1).
It is valuable remarking that CMor and gauss1 are transform-invariant wavelets; hence achieve isoresolution \cite{Soares}, 
yielding a balanced time and frequency entropy (see Proposition 4). Furthermore, Gabor inequality established a 
lower bound on the product between the variances associated with these two densities. 
There might thus exist some uncertainty principle between the time and the frequency entropy, 
i.e., a lower bound on $H_t (\psi).H_f (\psi)$.

\begin{proposition}:  
The uncertainty principle for Shannon entropy of a wavelet family is given by $H_t (\psi). H_f (\psi) \geqslant  H_{\psi}^{2} /4$.
\end{proposition}

\begin{proof}
The minimum of the product $H_{t}(\psi ).\left ( H_{\psi }-H_{t}(\psi ) \right )$ is achieved when $H_t (\psi)= H_{\psi} /2$ 
so that $H_t (\psi)= H_f (\psi)$, establishing thereby the bound.
\end{proof}

The entropy of the wavelets presented in Table I was subsequently determined (Table II), 
displaying how much the Morlet wavelet is endowed with wavelet analysis. 
The closed expression for the entropy of the (complex) Morlet wavelet is

\begin{equation}
H_t (\psi _{CMor} ) = H_f (\psi _{CMor} )= log_2 (\sqrt{\pi e}),
\end{equation}

so that $H_{CMor} = log_2 (\pi e)$. The time entropy of the standard Haar wavelet is equal to unity as predicted by Proposition 2. 
Alike, the frequency entropy of the complex Shannon wavelet is bounded by $log_2(\pi)$

The frequency entropy of de Oliveira wavelet is bounded by $H_f(deO)<log_2(\pi+3 \pi \alpha)$, i.e., 
2.030008 and 2.329568 for $\alpha$=0.1 and 0.2, respectively 
(referenced at the URL of the author's Website http://www2.ee.ufpe.br/codec/WEBLET.html)

\begin{figure}[!htb]
\centering
\includegraphics[width=0.5\textwidth]{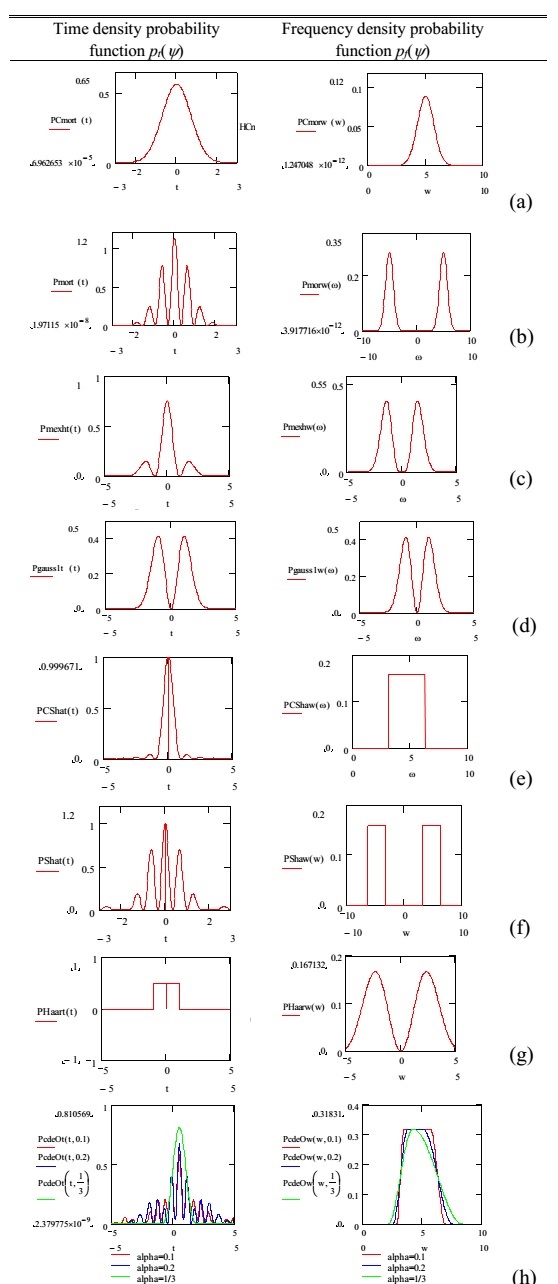}
\caption{Time and frequency probability density functions associated with some standard wavelets. Plots for: 
(a) Complex Morlet wavelet; (b) Real Morlet wavelet; (c) Mexican hat wavelet;  (d) gauss1 wavelet, (e) Complex Shannon wavelet; (f) Real Shannon wavelet; (g) Haar wavelet;  
(h) de Oliveira wavelet for different roll-off parameters, $\alpha$=0.1, 0.2, and 1/3 \cite{de Oliveira1}.}
\end{figure}

We speculate that the minimum entropy $H$ is achieved by a double Gaussian distribution (time and frequency domain), 
which is entirely in agreement with the long-standing concept of logon by Gabor \cite{Gabor}. 
In fact, a couple of wavelets are particularly significant: Morlet wavelet and Haar wavelet. 
Inquiringly, these were exactly pioneer wavelets! 

\begin{table}[!h]
\caption{ENTROPY OF SOME WAVELETS: TIME ENTROPY, FREQUENCY ENTROPY, AND AREA OF A WAVELET CELL IN THE JOINT t-f PLANE AND GLOBAL ENTROPY. WAVELETS: MORLET, SOMBRERO, SHANNON, GAUSS1, DE OLIVEIRA AND HAAR.} \label{tab1}
\begin{tabular}{c c c c c}
\hline \hline
CMor &1.547096 & 1.547096 & 2.393506 & 3.094191\\
Mor    & 1.104425&	2.547095& 2.813075 & 3.651520\\
mexh	&1.715098&	1.988567& 3.410587 & 3.703665\\
Sinc	&2.221052&	1.651383& 3.667807 & 3.872435\\
gauss1&1.937147&	1.937147& 3.752538 & 3.874293\\
Sha	&1.768634&	2.651665& 4.689824 & 4.420299\\
CdeO&&&\\
$\alpha$=0.1& 3.045824& 1.818698 & 5.539434& 4.864522\\
$\alpha$=0.2& 2.915276& 1.985887 &5.789408 &4.901163\\
Haar	&1.000000 & 3.985653 &3.985653 &4.985653\\
\hline
\end{tabular}
\end{table}

A good number of orthogonal wavelets cannot be described by analytical expressions, 
but fairly via filter coefficients [20] Daubechies, Symmlets, Coiflets, or even Mathieu wavelets \cite{Lira}. 
The Shannon entropy of such wavelets can be found by using the so-called two-scale relationship of a multiresolution analysis \cite{Mallat} \cite{Jawerth}:
Two-scale relation of the scaling function:
\begin{equation}
\phi (t)=\sqrt{2}\sum_{l=-\infty }^{+\infty } h_l \phi (2t-l)
\end{equation}
Two-scale relation of the wavelet:
\begin{equation}
\psi (t)=\sqrt{2}\sum_{l=-\infty }^{+\infty } g_l \phi (2t-l)
\end{equation}

Let $\mathbb{Z}$ denote the set of integers. The low-pass $H(.)$ filter of the MRA: 
$H(\omega ):=\sum_{l\epsilon \mathbb{Z}}{} h_l e^{-j\omega l}$ with $H(0)=\sqrt2$ and $H( \pi)=0$.
The high-pass $G(.)$ filter of the MRA:

$G(\omega ):=\sum_{l\epsilon \mathbb{Z}}{} g_l e^{-j\omega l}$ with $G(0)=0$ and $G( \pi)=\sqrt2$.
For the sake of simplicity, filter coefficients will encompass the term $\sqrt2$, i.e., 
$h_{k}\leftarrow \frac{1}{\sqrt2} h_{k}$ and $g_{k}\leftarrow \frac{1}{\sqrt2} g_{k}$. In these cases,
\begin{equation}
\int_{-\infty }^{+\infty }\phi ^2(t)dt=\sum_{k\epsilon \mathbb{Z}}{} h_{k}^{2}=E_{\phi }=1
\end{equation}

\begin{equation}
\int_{-\infty }^{+\infty }\psi ^2(t)dt=\sum_{k\epsilon \mathbb{Z}}{} g_{k}^{2}=E_{\psi }=1.
\end{equation}

The (discrete) probability density of $\psi$  is described by the probabilities $\left \{ g_{k}^{2} \right \}_{k\epsilon \mathbb{Z}}$. 
The MRA wavelet entropy can be computed by

\begin{equation}
H_{MRA}(\psi):=\sum_{k=-\infty }^{+\infty } g_{k}^{2}log_2 (g_{k}^{2})=\sum_{k=-\infty }^{+\infty } h_{k}^{2}log_2 (h_{k}^{2})
\end{equation}

\begin{table}[!h]
\caption{TIME ENTROPY OF SOME ORTHOGONAL WAVELETS DEFINED BY MRA FILTERS. 
WAVELETS OF COMPACT SUPPORT: DAUBECHIES, SYMMLET WAVELETS WITH $N=1,..,4$ VANISHING MOMENTS; 
OTHER WAVELETS: COIFLETS $N=1, 2$ AND ELLIPTIC-CYLINDRICAL MATHIEU WAVELET WITH PARAMETERS $\nu$ AND $q$ .} \label{tab1}
\begin{tabular}{c c c c c}
\hline \hline
          & Daubechies 	& Symmlet 		& Coiflet 		& Mathieu	                    \\
N        & HMRA(dBN) 	& HMRA(symN)	& HMRA(coifN) 	& $\nu$; $q$   HMRA(mth)\\
1	&1.000000		&1.000000		& 1.183011		&  1; 5     2.097353\\
2	&1.165857		&1.165857		& 1.376543		&  5; 5     1.739816\\
3	&1.397665		&1.397665 		&                            	&                            \\				
4	&1.447745		&1.345513     	&		           &                            \\
\hline
\end{tabular}
\end{table}

Let $H_2 (p):=-p.log_2 p-(1-p).log_2 (1-p)$ be the Shannon binary entropy \cite{Shannon}. 
The Shannon entropy of the dB1 wavelet (Haar), which is described by 
$\left [    \frac{1}{\sqrt2} \frac{1}{\sqrt2}  \right ]$, is 
$H_{MRA} ({\psi}_{dB1} )= H ({\psi}_{Haar})= H_2 (1/2)=1$ as expected. 
For dB2, whose filter coefficients are given by

\begin{equation}
\left [ \frac{1+\sqrt{3}}{4\sqrt2}, \frac{3+\sqrt{3}}{4\sqrt2}, \frac{3-\sqrt{3}}{4\sqrt2}, \frac{1-\sqrt{3}}{4\sqrt2}\right ]
\end{equation}

the entropy is $H_{MRA} (\psi _{dB2} )=1.16585703$. Values of the Shannon entropy for a few orthogonal wavelets 
defined by MRA filters are presented in the Table III. 
The bound derived for supportly compact wavelets (Proposition 2) can be checked without effort. 
We hypothesize from Eqn(18) and Table III that $H_{MRA} (\psi)$ is probably the same as $H_t (\psi)$.

\section{Concluding Remarks}

Wavelet is a body of knowledge of enormous fascination and far-reaching utility in signal processing, 
which is advancing at an astonishing pace. Mimicking the quantum mechanics approach, each of (continuous) wavelets is associated with 
two probability density functions one on the time domain, and another on the frequency domain. 
Consequently, the Shannon entropy and Renyi entropy of a wavelet were defined. 
We derived some sort of entropy conservation principle, which stated that wavelet versions resulting from the same mother wavelet 
retain the same Shannon entropy. The logical consistency of the m$\grave{e}$lange of claims throughout this paper provides 
certain evidence on the worthiness of our approach. 
The emphasis of this paper was on conveying the chief ideas as opposed to presenting a formal mathematical development or applications. 
Indeed, much was left to be investigated. Despite the fact that a mere trough draft of this technique had been outlined, 
it instigates certain expectation on both theoretical and practical information theory outcome related to wavelet analysis 
(e.g. wavelet compression can be carried out on wavelet-based information-theory-oriented algorithms not on the energy.) 
We foresee that our attempt to present the underlying philosophy of behind the wavelet information theory may help users navigate the ocean of wavelets.

\section{ACKNOWLEDGMENTS}
The author thanks Dr. Renato Jos\'e de Sobral Cintra (DE-UFPE) for quite a lot of comments. Worth also a thank an anonymous reviewer for perusal and comments.

\end{document}